\documentclass[conference]{IEEEtran}

\usepackage{common}
\usepackage{mymacro}
\usepackage[hide]{notes}
\allowdisplaybreaks

\usepackage{amsmath,amssymb,amsfonts}
\usepackage{algorithmic}
\usepackage{graphicx}
\usepackage{textcomp}
\usepackage{xcolor}
\def\BibTeX{{\rm B\kern-.05em{\sc i\kern-.025em b}\kern-.08em
    T\kern-.1667em\lower.7ex\hbox{E}\kern-.125emX}}

\newif\ifsupp 
\supptrue 
\newif\ifmain
\maintrue

\newtheorem{problem}{Problem}

\newtheorem{theorem}{Theorem}
\newtheorem{lemma}[theorem]{Lemma}
\newtheorem{corollary}[theorem]{Corollary}

\begin{document}

\title{
Coresets remembered and items forgotten:\\
submodular maximization with deletions
}

\author{\IEEEauthorblockN{Guangyi Zhang}
\IEEEauthorblockA{KTH Royal Institute of Technology \\
\textit{Division of Theoretical Computer Science} \\
Sweden \\
\href{mailto:guaz@kth.se}{guaz@kth.se}}
\and
\IEEEauthorblockN{Nikolaj Tatti}
\IEEEauthorblockA{HIIT, University of Helsinki \\
\textit{Department of Computer Science} \\
Finland \\
\href{mailto:nikolaj.tatti@helsinki.fi}{nikolaj.tatti@helsinki.fi}}
\and
\IEEEauthorblockN{Aristides Gionis}
\IEEEauthorblockA{KTH Royal Institute of Technology \\
\textit{Division of Theoretical Computer Science} \\
Sweden \\
\href{mailto:argioni@kth.se}{argioni@kth.se}}
}

\maketitle

\ifmain 
\begin{abstract}
In recent years we have witnessed an increase on the development of methods for submodular optimization, 
which have been motivated by the wide applicability of submodular functions in real-world data-science problems.
In this paper, we contribute to this line of work by considering the problem of 
\emph{robust submodular maximization against unexpected deletions},
which may occur due to privacy issues or user preferences.
Specifically, we consider the minimum number of items an algorithm has to remember,
in order to achieve a non-trivial approximation guarantee 
against adversarial deletion of up to \nD items.
We refer to the set of items that an algorithm has to keep before adversarial deletions 
as a \emph{deletion-robust coreset}.

Our theoretical contributions are two-fold.
First, we propose a \emph{single-pass streaming algorithm} that 
yields a $(1-2\Cgap)/(4\nm)$-approximation for maximizing a non-decreasing submodular function under a general \nm-matroid constraint and 
requires a coreset of size $\nI + \nD/\Cgap$, 
where \nI is the maximum size of a feasible solution.
To the best of our knowledge,
this is the first work to achieve an (asymptotically) \emph{optimal} coreset,
as no constant-\-factor approximation is possible with a coreset of size sublinear in $\nD$.
Second, we devise an effective offline algorithm that guarantees stronger approximation ratios with a coreset of size $\bigO(\nD \log(\nI)/\Cgap)$.
We also demonstrate the superior empirical performance of the proposed algorithms in real-life applications.
\end{abstract}


\section{Introduction}
\label{section:intro}

Submodular maximization has attracted much attention in the data-science community in recent years.
Its popularity is due to the ubiquity of the ``diminishing-returns'' property in different problem settings
and the rich toolbox that has been developed during the past decades~\citep{krause2014submodular}.
The problem of maximizing a non-decreasing submodular function can be used to cast a wide
range of applications, including
viral marketing in social networks \citep{kempe2015maximizing},
data subset selection \citep{wei2015submodularity}, and
document summarization \citep{lin2011class}.

However, in a world full of uncertainty, 
a pre-computed high-quality solution may cease to be feasible due to unexpected deletions.

As an example, consider an application in movie recommendation, 
where we ask to select a subset of movies to recommend to a user so as to maximize a certain
non-decreasing submodular utility function.
It is possible that the user has already seen some or all movies in the recommended set, 
and the rest are insufficient for providing a good-quality recommendation.
Such unexpected deletions may happen in other scenarios, for example, 
a user may exercise their ``right to be forgotten'', specified by EU's 
General Data Protection Regulation (GDPR)~\citep{voigt2017eu}, 
and request at any point certain data to be removed. 

Instead of re-running the recommendation algorithm over the dataset excluding the deleted items,
which typically is time-consuming, 
a better way to handle unexpected deletions is to extract a small and robust \emph{coreset} of the dataset, 
from which we can quickly select a new solution set.
To quantify robustness, we require the coreset to be robust against adversarial deletions up to \nD items.
Naturally, the coreset size has to depend on the number of deletions \nD.

Different adversarial models have been studied in the literature. 
For the most powerful adversary, called an \emph{adaptive} adversary, 
we assume that the coreset is known to the adversary and deletions occur in a worst-case manner.
Such malicious deletions lead to weak quality guarantees \citep{mitrovic2017streaming}, 
or require a larger coreset and longer running time \citep{mirzasoleiman2017deletion}.

In many applications, though, 
deletions are typically more benign and may only mildly corrupt the coreset.
For example, in the movie-recommendation scenario, 
the probability that a user has watched a movie may depend on the popularity of the movie, 
and be independent on whether the movie has been added in a coreset.
An adversary who is oblivious to the contents of the coreset is called a \emph{static} adversary.

In this paper, we derive bounds for the smallest possible coreset size that suffices to provide 
a non-trivial quality guarantee against item deletions by a static adversary.
It is known that no constant-factor approximation is possible with a coreset whose size is sublinear in $\nD$ \citep{dutting2022deletion}.
We assume that the adversary has unlimited computational power, 
and knows everything about the algorithm except for its random bits.

More concretely, given a non-decreasing submodular function, we study the problem 
\emph{Robust Coreset for submodular maximization under Cardinality constraint} 
(\rcc) against a static adversary.
In addition, we study the generalization of \rcc over a more general \nm-matroid constraint (\rcp), 
or a \nm-system constraint (\rcpsystem).
Recall that a cardinality constraint is known as a uniform matroid.
Throughout the paper, we are mostly interested in the more challenging case 
where $\nI = \bigO(\nD)$ and \nI is the maximum size of a feasible solution.

\begin{table*}[t]
\caption{A summary of existing results on robust coresets.
For simplicity, it is assumed that $\nI = \bigO(\nD)$.}
\label{tbl:results}
\small
\centering
\begin{tabular}{lccccc}
\toprule
	&\makecell{Adversarial\\model}	&\makecell{Constraint\\type}	&\makecell{Approximation\\factor} 	&\makecell{Coreset\\size}	&\makecell{Streaming/\\Offline}\\
\midrule
\citet{mitrovic2017streaming}	&adaptive	&cardinality	&$\approx0.149$	&$\bigO(\nD\log^3(\nI)/\Cgap)$	&S\\
\citet{mirzasoleiman2017deletion}; \citet{badanidiyuru2014streaming}	&adaptive	&cardinality	&$(1-\Cgap)/2$	&$\bigO(\nD\nI/\Cgap)$	&S\\
\citet{mirzasoleiman2017deletion}; \citet{chakrabarti2015submodular}	&adaptive	&\nm-matroids	&$1/(4\nm)$	&$\bigO(\nD\nI)$	&S\\
\citet{kazemi2018scalable}	&static	&cardinality	&$(1-\Cgap)/2$	&$\bigO(\nD\log^2(\nI)/\Cgap^3)$	&S\\
\rowcolor{gray!25}
This paper	&static	&cardinality	&$(1-\Cgap)/2$	&$\bigO(\nD \log(\nI)/\Cgap^2)$	&S\\
\citet{dutting2022deletion}	&static	&matroid	&$\frac{1-\Cgap}{e/(e-1) + 4 + \bigO(\Cgap)}$	&$\bigO(\nD\log(\nI/\Cgap)/\Cgap^2)$	&S\\
\rowcolor{gray!25}
This paper	&static	&\nm-matroids	&$(1-2\Cgap)/(4\nm)$	&$\mathbf{\nI+\nD/\Cgap}$	&S\\
\midrule
\citet{feldman2020one}	&adaptive	&cardinality	&$0.514$	&$\bigO(\nD\nI)$	&O\\
\citet{kazemi2018scalable}	&static	&cardinality	&$(1-\Cgap)/2$	&$\bigO(\nD\log(\nI)/\Cgap^2)$	&O\\
\rowcolor{gray!25}
This paper	&static	&cardinality	&$(1-\Cgap)/2$	&$\bigO(\nD \log(\nI)/\Cgap)$	&O\\
\citet{dutting2022deletion}	&static	&matroid	&$\frac{1-\Cgap}{e/(e-1) + 2 + \bigO(\Cgap)}$	&$\bigO(\nD\log(\nI/\Cgap)/\Cgap^2)$	&O\\ 
\rowcolor{gray!25}
This paper	&static	&matroid	&\cleanfactorrcm	&$\bigO(\nD \log(\nI)/\Cgap)$	&O\\
\rowcolor{gray!25}
This paper	&static	&\nm-system	&\cleanfactorrcp	&$\bigO(\nD \log(\nI)/\Cgap)$	&O\\
\bottomrule
\end{tabular}
\end{table*}

Our contributions are summarized as follows.
\begin{itemize}
	\item We offer a randomized $\frac{1-2\Cgap}{4\nm}$-approximation, single-pass streaming algorithm for the \rcp problem, with a coreset of size $\nI + \nD/\Cgap$.
	The coreset size is asymptotically optimal.
	Prior to our work, the best-known coreset size is $\bigO(\nD \log(\nI)/\Cgap^2)$ 
	(more details about related work are shown in Table~\ref{tbl:results} and discussed 
	in Section~\ref{section:related}).
		
	\item In addition, 
	we introduce and analyze a natural greedy algorithm, 
	which keeps multiple backups for each selected item.
	We show that this algorithm offers stronger approximation ratios at the expense of a larger coreset size.
	Specifically, we devise an offline algorithm that requires a coreset of size 
	$\bigO(\nD \log(\nI)/\Cgap)$ and yields
	$\frac{1-\Cgap}{2}$ and 
	\cleanfactorrcp
	approximation for \rcc and \rcpsystem, respectively.
	Besides, the greedy algorithm is empirically effective even against an adaptive adversary.
	
	\item The proposed algorithms are evaluated empirically and are shown to 
	achieve superior performance in many application scenarios.
\end{itemize}

Our techniques can be extended to obtain a $(1-\Cgap)/2$-approximation one-pass streaming algorithm 
for the \rcc problem with a coreset of size $\bigO(\nD \log(\nI)/\Cgap^2)$, 
an improvement over $\bigO(\nD \log(\nI)^2/\Cgap^3)$ in \citet{kazemi2018scalable}.
Our approximation ratio for the \rcp problem can be further strengthened to \cleanfactorrcm when the constraint is a single matroid.
A summary of the existing results is displayed in Table \ref{tbl:results}.
Our implementation can be found at \code.

The rest of the paper is organized as follows.
We formally define the robust coreset problem in Section~\ref{section:definition}.
Related work is discussed in Section~\ref{section:related}.
The proposed streaming and offline algorithms are presented and analyzed in Sections~\ref{section:streaming} and~\ref{section:offline}, respectively.
Our empirical evaluation is conducted in Section~\ref{section:experiment},
followed by a short conclusion in Section~\ref{section:conclusion}.

\section{Problem definition}
\label{section:definition}

In this section we define the concept of \emph{robust coreset} that we consider in this paper.
Before discussing the concept of coreset and formally define the problem we study, 
we briefly review the definitions of submodularity, \nm-matroid, and \nm-system.

\spara{Submodularity.}
Given a set \V, a function $f: 2^\V \to \reals_+$ is called \emph{submodular} if 
for any $X \subseteq Y \subseteq \V$ and $v \in \V \setminus Y$, it holds $f(v \mid Y) \le f(v \mid X)$,
where $f(v \mid Y) = f(Y+v) - f(Y)$ is the \emph{marginal gain} of $v$ with respect to set $Y$.
Function $f$ is called \emph{non-decreasing} if 
for any $X \subseteq Y \subseteq \V$, it holds $f(Y) \ge f(X)$.
Without loss of generality, we can assume that the function $f$ is normalized, i.e., $f(\emptyset)=0$.

\spara{\nm-matroid.}
For a set \V,
a family of subsets $\M \subseteq 2^\V$ is called \emph{matroid} if it satisfies the following two conditions:
(1)~downward closeness: if $X \subseteq Y$ and $Y \in \M$, then $X \in \M$;
(2)~augmentation: if $X,Y \in \M$ and $|X| < |Y|$, then $X + v \in \M$ for some $v \in Y \setminus X$.
For a constant~\nm, 
a \emph{\nm-matroid} $\M \subseteq 2^\V$ is defined as the intersection 
of \nm matroids $\{\M_j\}_{j \in [\nm]}$.
The rank of a \nm-matroid \M is defined as $\nI = \max_{S \in \M} |S|$.

\spara{\nm-system.}
For a set \V and a constant \nm, a \nm-system $\M \subseteq 2^\V$ is defined as follows.
Given a set $Y \subseteq \V$, a set $X$ is called a \emph{base} of $Y$ if $X$ is a maximal subset of $Y$, i.e.,
$X \in \M$, $X \subseteq Y$ and $X+v \notin \M$ for any $v \in Y \setminus X$.
We denote the set of bases of $Y$ by $\B(Y)$.
A tuple $(\V,\M)$ forms a \nm-system if for any $Y \subseteq \V$, it is 
$\frac{\max_{X \in \B(Y)} |X|}{\min_{X \in \B(Y)} |X|} \le \nm$.
A \nm-matroid is a special case of a \nm-system.

\spara{Robust coreset.}
We consider a set \V, 
a non-decreasing submodular function $f: 2^\V \to \reals_+$,
an integer $\nD$, and
a \nm-matroid $\M \subseteq 2^\V$ of rank $\nI$.
A procedure for selecting a solution subset for $f$, 
which is robust under deletions, is specified by the following three stages.
\begin{enumerate}
	\item Upon receiving all items in \V, an algorithm $\alg_1$ returns a small subset $\R \subseteq \V$ as the coreset.
	\item A static adversary deletes a subset $\D \subseteq \V$ of size at most~\nD.
	\item An algorithm $\alg_2$ extracts a feasible solution $\I \subseteq \R \setminus \D$ and $\I \in \M$.
\end{enumerate}
The algorithms $\alg_1$ and $\alg_2$ can be randomized.
We call an adversary \emph{static} if the adversary is unaware of the random bits used by the algorithm.
Thus, one can assume the set of deletions \D are fixed before the algorithm is run.
The quality of the solution $\ALG = \I$ is measured by evaluating the function $f$ on the solution set \ALG.
We aim that the quality $f(\ALG)$ is close to the optimum after deletion, 
$f(\OPT(\D))$, where $\OPT(\D) = \arg\max_{S \subseteq \V \setminus \D, S \in \M} f(S)$.
We omit \D in $\OPT(\D)$ when it is clear from the context.
We say that a pair of randomized algorithms $(\alg_1,\alg_2)$ yield a \emph{$(\Capx,\nR)$-coreset} if
\begin{equation}
\label{eq:coreset}
\expect[f(\alg_2(\ret(\alg_1), \D))] \ge \Capx f(\OPT)
\ \ \text{and} \ \
|\R(\alg_1)| \le \nR,
\end{equation}
where algorithm $\alg_1$ returns a tuple of sets $\ret(\alg_1) = \{S_i\}$, and
$\R(\alg_1) = \cup_i S_i$ is the coreset.

\begin{problem}[Robust coreset for submodular maximization under \nm-matroid constraint (\rcp)]
Given
a set~\V, 
a non-decreasing submodular function $f: 2^\V \to \reals_+$,
a \nm-matroid~\M of rank~\nI, and
an unknown set~\D,
find an $(\Capx,\nR)$-coreset \R.
\label{problem:rcp}
\end{problem}
Notice that there is a trade-off between the approximation ratio~$\Capx$ 
and the coreset size \nR.
We typically aim for \nR to be independent of $|\V|$ and grow slowly in \nD.
Note that even in the case $\D = \emptyset$, i.e., no deletions,
extracting an optimal solution from the items of \V for the \rcp problem is an intractable problem.
In fact, no polynomial-time algorithms can approximate $f(\OPT)$ with a factor 
better than $1-1/e$~in offline computation~\citep{nemhauser1978best}, 
or better than $1/2$ in a single pass~\citep{feldman2020one}.

We refer to Problem~\ref{eq:coreset} with 
cardinality constraint, 
single matroid constraint, and
\nm-system constraint, 
as \rcc, \rcm, and \rcpsystem, respectively.

\section{Related work}
\label{section:related}

\para{Deletion-robust submodular maximization.}
For simplicity, we assume $\nI = \bigO(\nD)$.
\citet{mitrovic2017streaming} study robust submodular maximization against an adaptive adversary who can inspect the coreset before deletion.
They give a one-pass constant-approximation algorithm for the \rcc problem 
with a coreset of size  $\bigO(\nD\log^3(\nI)/\Cgap)$.
\citet{mirzasoleiman2017deletion} provide another simple and flexible algorithm, 
which sequentially constructs $\nD+1$ solutions by running any existing streaming algorithm.
This gives an $1/(4\nm)$-approximation algorithm 
for the \rcp problem in a single pass at the expense of a coreset of larger size $\bigO(\nD \nI)$.
When an offline coreset procedure is allowed, \citet{feldman2020one} propose a 0.514-approximation algorithm for the \rcc problem with a coreset of size $\bigO(\nD \nI)$ by using a two-player protocol.

For a static adversary, \citet{kazemi2018scalable} achieve $(1-\Cgap)/2$ approximation for the \rcc problem with coreset size 
$\bigO(\nD\log(\nI)/\Cgap^2)$ and
$\bigO(\nD\log^2(\nI)/\Cgap^3)$
for offline and one-pass streaming settings, respectively.
Very recently, \citet{dutting2022deletion} generalize the work of \citet{kazemi2018scalable} into a matroid constraint, requiring a coreset size of $\bigO(\nD\log(\nI/\Cgap)/\Cgap^2)$.

Compared to prior work, our algorithm achieves the best-known coreset size against a static adversary.

\spara{Max-min robust submodular maximization.}
Another line of work on robust submodular maximization studies a different notion of robustness, where adversarial deletions are performed directly on the \emph{solution} and no further updates to the solution are allowed \citep{krause2008robust,orlin2018robust,bogunovic2017robust}

\spara{Dynamic submodular maximization.}
The input in the dynamic model consists of a stream of updates, 
which could be either an insertion or a deletion of an item.
It is similar to the \rcc setting if all deletions arrive at the end of the stream.
However, the focus in the dynamic model is time complexity instead of space complexity.
Methods aim to maintain a good-quality solution at any time with a small amortized update time \citep{lattanzi2020fully,monemizadeh2020dynamic,chen2021complexity}.

\spara{Submodular maximization.}
The first single-pass streaming algorithm for a cardinality constraint proposed by \citet{badanidiyuru2014streaming} relies on a thresholding technique.
This simple technique turns out to yield tight 1/2-approximation unless the memory depends on \nV \citep{feldman2020one,norouzi2018beyond}.
The memory requirement is later improved from $\bigO(\nI \log (\nI)/\Cgap)$ to $\bigO(\nI/\Cgap)$~\citep{kazemi2019submodular}.
For a general \nm-matroid constraint, $1/4\nm$-approximation has been known \citep{chekuri2015streaming,chakrabarti2015submodular}.

In the offline setting, 
it is well-known that a simple greedy algorithm achieves optimal $1-1/e$ 
approximation for a cardinality constraint \citep{nemhauser1978best,nemhauser1978analysis}. 
A modified greedy algorithm obtains $1/(\nm+1)$-approximation for a general \nm-system constraint~\citep{fisher1978analysis,calinescu2011maximizing}.
More sophisticated algorithms with tight $(1-1/e)$-approximation for a matroid appeared later \citep{calinescu2011maximizing,filmus2012tight}.

\section{The proposed streaming algorithm}
\label{section:streaming}

In this section, we first describe a non-robust streaming algorithm \exc, 
and then introduce a novel method that enhances it for the \rcp problem.
We also discuss an improved streaming algorithm for the simpler \rcc problem.

\citet{chakrabarti2015submodular} provide a simple $1/(4\nm)$-approximation streaming algorithm for non-decreasing submodular maximization under a \nm-matroid constraint.
We call their algorithm \exc because it maintains one feasible solution at all time by exchanging cheap items $\Sw$ in the current solution \I,  
for any new valuable item $v$ that cannot be added in \I without making the solution infeasible.
The \exc algorithm measures the value of an item by a weight function $\val: \V \to \reals_+$, 
which is defined as $\val(v) = f(v \mid \I_v)$, 
where $\I_v$ is the feasible solution before processing item $v$.
Furthermore, the algorithm measures the value of a subset
by extending \val with $\val(S) = \sum_{u \in S} \val(u)$.
The algorithm replaces $\Sw$ with $v$ in $\I$ when $\val(v) \ge (1+\Cexc) \val(\Sw)$,
for a 
parameter~\Cexc.
We restate the \exc algorithm of \citet{chakrabarti2015submodular} in Algorithm \ref{alg-exc-nonrobust} 
and their main result in Theorem~\ref{theorem:exc-nonrobust}.
We stress that the \exc algorithm is non-robust and Theorem~\ref{theorem:exc-nonrobust} holds only in the absence of deletions \D.

\begin{algorithm}[t]
\DontPrintSemicolon
\KwIn{parameter \Cexc}
\SetKwFunction{Fexchange}{Exchange}
$\I \gets \emptyset$\;
\For{$v \in \V$}{
	$\val(v) \gets f(v \mid \I)$\;
	$\Sw \gets \Fexchange(v, \I)$\;
	\If{$\val(v) \ge (1+\Cexc) \val(\Sw)$}{
		$\I \gets \I + v - \Sw$\;
	}
}
\Return{$\I$}\;
\;
\Function{\Fexchange$(v, \I)$}{
	\For{$j \in [\nm]$}{
		\If{$\I + v \not\in \M_j$}{
			$u_j \gets \arg\min_{u \in \I: \I + v - u \in \M_j} \val(u)$
		}
	}
	\Return{$\{ u_j \}_{j \in [\nm]}$}\;
}
\caption{\exc streaming algorithm in Chakrabarti and Kale (2015)}
\label{alg-exc-nonrobust}
\end{algorithm}

\begin{theorem}[\citet{chakrabarti2015submodular}]
\label{theorem:exc-nonrobust}
Suppose Algorithm~\ref{alg-exc-nonrobust} is run over items \V.
For any $\Cexc > 0$ and any feasible solution $S \subseteq \V$ under a \nm-matroid constraint, 
Algorithm~\ref{alg-exc-nonrobust} returns a feasible solution \I that satisfies 
\[
	f(S) \le \Capxv \val(\I) \leq \Capxv f(\I),
\]
where $\Capxv = (\nm(\Cexc+1) - 1) (\Cexc+1) / \Cexc + 1+1/\Cexc$.
In particular, when $\Cexc = 1$, we have
$f(S) \le 4\nm\, \val(\I) \leq 4\nm\,  f(\I)$.
\end{theorem}

\medskip
In this paper, we develop a robust extension of the \exc algorithm (\rexc), 
displayed in Algorithms~\ref{alg-exc} and~\ref{alg-exc-after}.
Algorithm~\ref{alg-exc} simply inserts a randomized buffer \C between the data stream and the \exc algorithm, and
Algorithm~\ref{alg-exc-after} continues to process items in $\C \setminus \D$ after deletions and returns a final solution.
Note that Algorithms~\ref{alg-exc} and~\ref{alg-exc-after} can be seen as two stages of the \exc algorithm.
A parameter \Cgap is used to set the size of the buffer \C to $\nD/\Cgap$.
A smaller value of \Cgap and a larger buffer lead to a stronger robust guarantee.

\begin{algorithm}[t]
\DontPrintSemicolon
\KwIn{parameter \Cgap, \Cexc}
$\I \gets \emptyset, \C \gets \emptyset$\;
\For{$v' \in \V$}{
	$\C \gets \C + v'$\;
	\If{$|\C| \ge \nD/\Cgap$}{
		Sample and remove an item $v$ from \C with probability proportional to $1 / f(v \mid \I)$\;
		$\val(v) \gets f(v \mid \I)$\;
		$\Sw \gets \Fexchange(v, \I)$\;
		\If{$\val(v) \ge (1+\Cexc) \val(\Sw)$}{
			$\I \gets \I + v - \Sw$
			\label{step:accept1} \;
		}
	}
}
\Return{\I and \C}\;
\caption{Robust \exc streaming algorithm (\rexc)}
\label{alg-exc}
\end{algorithm}

\begin{algorithm}[t]
\DontPrintSemicolon
\KwIn{\I, \C, \D, and parameter \Cexc}
\For{$v \in \C \setminus \D$}{
	$\val(v) \gets f(v \mid \I)$\;
	$\Sw \gets \Fexchange(v, \I)$\;
	\If{$\val(v) \ge (1+\Cexc) \val(\Sw)$}{
		$\I \gets \I + v - \Sw$
		\label{step:accept2} \;
	}
}
\Return{$\I \setminus \D$}\label{step:return}
\caption{Construction of \rexc solution after deletions}
\label{alg-exc-after}
\end{algorithm}

It is easy to see that Algorithm~\ref{alg-exc} requires a coreset of size at most $\nI + \nD/\Cgap$ and 
at most $\bigO(\nV \nD/\Cgap)$ queries to function $f$, where $\nV = |\V|$.
Besides, it successfully preserves almost the same approximation guarantee as the non-robust \exc algorithm in the presence of adversarial deletions.
\begin{theorem}
\label{theorem:exc}
For any $\Cexc > 0$,
Algorithms~\ref{alg-exc} and~\ref{alg-exc-after} 
yield an approximation guarantee ${\left(1-(1+1/\Cexc)\Cgap\right)}/{\Capxv}$ 
for the \rcp problem
using a coreset of size $\nI + \nD/\Cgap$,
where $\Capxv = (\nm(\Cexc+1) - 1) (\Cexc+1) / \Cexc + 1+1/\Cexc$.
In particular, when $\Cexc = 1$, we obtain a
$\frac{1-2\Cgap}{4\nm}$-approximation guarantee.
\end{theorem}

For the simpler cardinality constraint, we can obtain a tighter approximation ratio at the expense of a larger coreset size $\bigO(\nD \log(\nI)/\Cgap^2)$.
This is an improvement over the state-of-the-art $\bigO(\nD \log^2(\nI)/\Cgap^3)$ in \citet{kazemi2018scalable}.
The main idea is the utilization of importance sampling (Lemma~\ref{lemma:imp}) on top of the robust \sieve algorithm in \citet{kazemi2018scalable}.
We defer the details to Section~\ref{section:streaming-cardinality} in Appendix~\cite{appendix}.
\begin{theorem}
\label{theorem:streaming-cardinality}
There exists a one-pass streaming algorithm that
yields $(1-2\Cgap)/2$ approximation guarantee for the \rcc problem, 
with a coreset size $\bigO((\nD/\Cgap + \nI) \log(\nI)/\Cgap)$.
\end{theorem}

In the rest of this section,
we prove Theorem~\ref{theorem:exc}.

\subsection{Proof of Theorem~\ref{theorem:exc}}

As we mentioned before,
Algorithms~\ref{alg-exc} and~\ref{alg-exc-after} can be seen as two stages of the non-robust \exc algorithm with input $(\V \setminus \C) + (\C \setminus \D)$.
To be more specific, first, 
Algorithm~\ref{alg-exc} finds a solution $S_1$ by running the \exc algorithm with input~$\V \setminus \C$.
Then, Algorithm~\ref{alg-exc} returns solution $S_1$ and buffer \C.
Then, Algorithm~\ref{alg-exc-after} finds a solution $S_2$ by processing 
the items that are preserved in $\C \setminus \D$, while starting from feasible solution~$S_1$.
Finally, Algorithm~\ref{alg-exc-after} returns solution $\ALG = S_2 \setminus \D$.
 
By Theorem~\ref{theorem:exc-nonrobust}
we know that the solution $S_2$ returned by Algorithm~\ref{alg-exc-after} 
before deletions occur is provably good.
This observation is formally stated in the following corollary.

\begin{corollary}
\label{corollary:exc-non-robust}
For any $\Cexc > 0$,
the feasible solution $S_2$ returned by Algorithm~\ref{alg-exc-after},
before the deletion of items in $S_2 \cap \D$, 
satisfies
\[
	f(\rOPT) \le \Capxv \val(S_2),
\]
where $\rOPT \in \M$ is the optimal solution over data $(\V \setminus \C) + (\C \setminus \D)$.
\end{corollary}

To complete the proof of Theorem~\ref{theorem:exc}, 
we need to show that the solution $S_2$ is robust against deletions, in expectation.
We first derive the expected loss in marginal gain of a sampled item by importance sampling due to the adversarial deletions.
Intuitively, among a candidate set of items with varied marginal gain, 
we need to down\-sample items with larger gain.
Otherwise, the adversary could target those items and we are likely to suffer a great loss.

\begin{lemma}\label{lemma:imp}
Consider sets $\C, S, \D \subseteq \V$. 
Define $\nD = |\D|$.
Let~$v \in \C$ be an item sampled with probability proportional to $1/ f(v \mid S)$.
Then the expected loss in marginal gain of the item $v$ after deleting \D is
$$
\expect\left[ f(v \mid S) \indicator[v \in \D] \right] \leq \frac{\nD}{|\C|} \expect[f(v \mid S)].
$$
\end{lemma}
\begin{proof}
We know
$$\expect[f(v \mid S)] = \sum_{v \in \C} f(v \mid S) p_v = |C|/z,$$
where $p_v = \frac{1/ f(v \mid S)}{z}$ and $z = \sum_{v \in \C} 1/ f(v \mid S)$.
If the sampled item $v$ is in $\D$, we suffer a loss of $f(v \mid S)$, and this happens with probability $p_v$.
Thus, the expected loss is $f(v \mid S) p_v = 1/z$.
That is to say, every item leads to the same amount of expected loss.
The expected loss after any deletion set is
\begin{align*}
&\expect[f(v \mid S) \indicator[v \in \D]] 
= \sum_{v \in \C} f(v \mid S) p_v \indicator[v \in \D] \\
&\quad= \sum_{v \in \C} \indicator[v \in \D]/z 
\leq \nD/z 
= \frac{\nD}{|\C|}\expect[f(v \mid S)]
\end{align*}
proving the claim.
\end{proof}

We proceed to show that solution $S_2$ is robust.
\begin{lemma}
\label{lemma:exc-robust}
For any $\Cexc > 0$,
given the feasible solution $S_2$ found by Algorithm~\ref{alg-exc-after} before removing items in $S_2 \cap \D$, 
we have
\[
	\expect[\val(S_2')] \ge (1-(1+1/\Cexc)\Cgap) \expect[\val(S_2)],
\]
where $S_2' = S_2 \setminus \D$.
\end{lemma}
\begin{proof}
Let \KI be the set of items that are ever accepted into the tentative feasible solution in Algorithms~\ref{alg-exc} and~\ref{alg-exc-after}, i.e., including $S_2$ and those that are first accepted but later swapped.
We first show that \KI is robust in the sense that
$\expect[\val(\KI')] \ge (1-\Cgap) \expect[\val(\KI)]$,
where $\KI' = \KI \setminus \D$.

Let $v_i$ be the $i$-th item added into \KI, where $i \le \nV = |\V|$.
We know that item $v_i$ is sampled from a candidate set $\C_i$ with a probability proportional to $1/\val(v_i)$.
Besides, $|\C_i| \ge \nD / \Cgap \ge |\D| / \Cgap$.
Thus,
\begin{align*}
	\expect[\val(\KI')] 
	&= \expect\big[\sum_{v \in \KI} \val(v) (1 - \indicator[v \in \D]) \big] \\
	&= \expect[\val(\KI)] - \expect\big[\sum_{v \in \KI} \val(v) \indicator[v \in \D]\big] \\
	&= \expect[\val(\KI)] - \sum_{i \le \nV} \expect\big[ \val(v_i) \indicator[v_i \in \D] \big] \\
	&\ge \expect[\val(\KI)] - \sum_{i \le \nV} \expect\big[ \frac{|\D|}{|\C_i|} \val(v_i) \big] \\
	&\ge \expect[\val(\KI)] - \sum_{i \le \nV} \expect\big[ \Cgap \val(v_i) \big] \\
	&= \expect[\val(\KI)] - \Cgap \expect[\val(\KI)],
\end{align*}
where the first inequality is due to Lemma~\ref{lemma:imp}.

Next, we show that $S_2$ is robust, too.
Let $\K = \KI \setminus S_2$, and $\K' = \K \setminus \D$.
A useful property about \K that is shown in \citet[Lemma~2]{chakrabarti2015submodular} is that $\val(S_2)/\Cexc \ge \val(\K)$.
Therefore,
\begin{align*}
	&\expect[\val(\K') + \val(S_2')] 
	 = \expect[\val(\KI')] \\
	&\qquad\ge (1-\Cgap) \expect[\val(\KI)] 
	= (1-\Cgap) \expect[\val(\K) + \val(S_2)].
\end{align*}
By linearity of expectation and rearranging, we have
\begin{align*}
	\expect[\val(S_2')] 
	&\ge (1-\Cgap) (\expect[\val(\K)] + \expect[\val(S_2)]) - \expect[\val(\K')] \\
	&\ge (1-\Cgap) (\expect[\val(\K)] + \expect[\val(S_2)]) - \expect[\val(\K)] \\
	&= (1-\Cgap) \expect[\val(S_2)] - \Cgap \expect[\val(K)] \\
	&\ge (1-\Cgap) \expect[\val(S_2)] - \Cgap \expect[\val(S_2)] /\Cexc \\
	&= (1-(1+1/\Cexc)\Cgap) \expect[\val(S_2)],
\end{align*}
completing the proof.
\end{proof}

Finally, we complete the proof of Theorem~\ref{theorem:exc}.
\begin{proof}[Proof of Theorem~\ref{theorem:exc}]
We know that \OPT is the optimal solution over data $\V \setminus \D$, which is worse than the optimum solution $O$ over $(\V \setminus \C) + (\C \setminus \D)$,
that is, $f(O) \ge f(\OPT)$.
Therefore,
\begin{align*}
	\expect[f(\ALG)] 
	&\ge \expect[\val(\ALG)] \\
	&= \expect[\val(S_2 \setminus \D)] \\
	&\ge (1-(1+1/\Cexc)\Cgap) \expect[\val(S_2)] 
	&& \triangleright\text{Lemma \ref{lemma:exc-robust}}\\
	&\ge (1-(1+1/\Cexc)\Cgap) f(O) /\Capxv
	&& \triangleright\text{Corollary~\ref{corollary:exc-non-robust}}\\
	&\ge (1-(1+1/\Cexc)\Cgap) f(\OPT) /\Capxv, 
\end{align*}
completing the proof.
\end{proof}

\section{The proposed offline algorithm}
\label{section:offline}
\begin{algorithm}[t]
\DontPrintSemicolon
\KwIn{parameter \Cgap}
$\R \gets $ the set of top-\nD items in \V according to $f(\{v\})$\;
$\V \gets \V \setminus \R, \, j \gets 1, \, \I_j \gets \emptyset$\;
\Do{$|\V| > 0$}{
	$C_j \gets$ top-$(\max \left\{\frac{\nD}{j\Cgap}, 1\right\})$ items in \V w.r.t. $f(v \mid \I_j)$\;
	$\R \gets \R \cup C_j$\;
	\If{$|\C_j| \ge \frac{\nD}{j\Cgap}$}{
		Sample an item $v_j$ from $\C_j$ with a probability proportional to $1 / f(v_j \mid \I_j)$\;
		$\I_{j+1} \gets \I_j + v_j$\;
	}
	$\V \gets \{v \in \V \setminus \C_j: \I_{j+1} + v \in \M \}, \, j \gets j+1$\;
}
\Return{$(\R, \{\I_j\}_j)$}\;
\caption{Offline robust coreset for \rcpsystem}
\label{alg-offline-coreset}
\end{algorithm}

\begin{algorithm}[t]
\DontPrintSemicolon
\KwIn{Coreset and auxiliary information $(\R, \{\I_j\}_j)$ returned by Algorithm~\ref{alg-offline-coreset}, set of deleted items $\D$}
$\I \gets \I_i$ where $i=\max_j j$ \;
$H \gets$ a greedy solution using items in $\R \setminus \D$ \label{step:greedy}\;
\Return{the best solution among $\{\I \setminus \D, H\}$}\;
\caption{Construction of \rcpsystem solution after deletion}
\label{alg-offline-selection}
\end{algorithm}

We start our exposition by presenting a unified Algorithm~\ref{alg-offline-coreset} to construct a robust coreset for both \nm-system and cardinality constraints.
However, different algorithms (Algorithms~\ref{alg-offline-selection} and~\ref{alg-offline-selection-cardinality}~\cite{appendix}, respectively) are needed to extract the final solution after the deletion of items by the adversary.

Algorithm~\ref{alg-offline-coreset} constructs a robust coreset by iteratively collecting 
the items with the largest marginal gains with respect to a tentative solution \I, 
and in each iteration, sampling an item from the collected set and adding it into \I.
Algorithm~\ref{alg-offline-selection} or~\ref{alg-offline-selection-cardinality} extracts the final solution after deletion.
The running time in terms of query complexity, 
i.e., the number of calls to function~$f$,
of Algorithms~\ref{alg-offline-coreset}, \ref{alg-offline-selection} and~\ref{alg-offline-selection-cardinality} is
$\bigO(\nV \nI)$,
$\bigO\!\left((\nD \log(\nI)/\Cgap + \nI) \nI\right)$, and
$\bigO\!\left((\nD \log(\nI)/\Cgap + \nI) \log(\nI)/\Cgap \right)$,
respectively.
Our main results are stated below.

\begin{theorem}
\label{theorem:offline-matroid}
Algorithms~\ref{alg-offline-coreset} and~\ref{alg-offline-selection} 
yield a $\frac{1}{\nm+1 + (\nm+1)/(1-\Cgap)}$ approximation guarantee for the \rcpsystem problem, 
using a coreset of size $\bigO(\nD \log(\nI)/\Cgap + \nI)$.
\end{theorem}
\note{Simplified ratio $\frac{1}{2(\nm+1) + 2(\nm+1)\Cgap}$ with $1/(1+\Cgap) \le 1+2\Cgap$ when $\Cgap < 0.5$.}

Using a proof similar to the one of Theorem~\ref{theorem:offline-matroid},
we can obtain a stronger approximation ratio for a single matroid constraint,
by replacing the greedy algorithm (Step~\ref{step:greedy}) in Algorithm~\ref{alg-offline-selection} 
by a more advanced continuous greedy algorithm \citep{calinescu2011maximizing}.
\begin{theorem}
\label{theorem:offline-1matroid}
Algorithms~\ref{alg-offline-coreset} and a modified Algorithm~\ref{alg-offline-selection} 
yield a $\frac{1}{e/(e-1) + 2/(1-\Cgap)}$ approximation for the \rcm problem, 
using a coreset of size $\bigO(\nD \log(\nI)/\Cgap + \nI)$.
\end{theorem}

In the simpler case of a cardinality constraint, we can achieve a better approximation ratio
by a \sieve-like algorithm \cite{badanidiyuru2014streaming} to extract the final solution.
\begin{theorem}
\label{theorem:offline-cardinality}
Algorithms~\ref{alg-offline-coreset} and~\ref{alg-offline-selection-cardinality}
yield a $(1-2\Cgap)/2$ approximation guarantee for the \rcc problem, 
using a coreset of size $\bigO(\nD \log(\nI)/\Cgap + \nI)$.
\end{theorem}

We will devote
the rest of this section for proving Theorem~\ref{theorem:offline-matroid}.
Proof for Theorem~\ref{theorem:offline-cardinality} is deferred to Appendix~\cite{appendix}.

\subsection{Proof of Theorem~\ref{theorem:offline-matroid}}
\label{section:offline-matroid}
The strategy in Algorithms~\ref{alg-offline-coreset} is to sample-and-keep disjoint candidate sets, which forces the adversary to invest its deletions among these disjoint sets.
To ensure a bounded expected loss due to the deletions, 
we perform importance sampling (also known as ``uselessness'' sampling) in Lemma~\ref{lemma:imp} among each candidate set.
We further show that it is safe to reduce the size of candidate sets harmonically, 
as the expected marginal gain of the sampled items is non-increasing.

For the remainder of the section, we will adopt the following notation. Let
$\{\I_i\}$ be the partial solutions discovered by Algorithm~\ref{alg-offline-coreset},
and let $v_i$ be the item added to $\I_i$, that is, $\I_{i + 1} = \I_i + v_i$.
Let $\C_i$ be the sets from which Algorithm~\ref{alg-offline-coreset} samples $v_i$.
In addition, let~$\D$ be the set of deleted items by the adversary.
Finally, we write $\I'_i = \I_i \setminus \D$.

Next we show that the gain of item $v_j$ is non-increasing in~$j$.
\begin{lemma}
\label{lemma:nonincgain}
For any $j < i$, we have
$f(v_j \mid \I_j) \ge f(v_{i} \mid \I_{i})$.
\end{lemma}
\begin{proof}
\deferredproof
\end{proof}

The following lemma shows the robustness of the tentative partial solution $\I$ built in Algorithm \ref{alg-offline-coreset}, in the sense that 
$\expect[f(\I'_i)]$ is close to $\expect[f(\I_i)]$.
Intuitively, the expected loss of the first item in \I is small as its candidate set $\C_1$ has a large size $\nD/\Cgap$.
A subsequent item in \I can be sampled with a decreasing candidate size, because previously added items can help compensate if its candidate set is attacked by the adversary.
\begin{lemma}\label{lemma:robust}
$\expect[f(\I_i')] \ge (1-\Cgap) \expect[f(\I_i)]$.
\end{lemma}
\begin{proof}
We start by bounding $\expect[f(\I'_i)]$,
\begin{align*}
\expect[f(\I'_i)] 
& = \expect[\sum_{j < i} f(v_j \mid \I_j \setminus D) \indicator[v_j \notin \D]] \\
&\ge \expect[\sum_{j < i} f(v_j \mid \I_j) \indicator[v_j \notin \D]] \\
&= \expect[f(\I_i)]  - \expect[\sum_{j < i} f(v_j \mid \I_j) \indicator[v_j \in \D]],
\end{align*}
where the inequality is due to submodularity.

Now we bound further the second term.
For simplicity let us write $g_j = f(v_j \mid \I_j)$.
Recall that $C_j$ is the set from which Algorithm~\ref{alg-offline-coreset} samples $v_j$.
Note that $|C_j| \geq \frac{\nD}{j \epsilon}$, and that the sets $\{C_j\}$ do not overlap. 
Define $D_j = C_j \cap D$.
Note that $\D_j$ is also a random variable like $\C_j$, 
which depends on previously sampled items $\I_j$.
Then
\begin{align*}
&\expect \Big[\sum_{j < i} g_j \indicator[v_j \in \D] \Big] 
= \sum_{j < i} \expect\big[ \expect\left[ g_j \indicator[v_j \in \D_j] \mid \I_j \right] \big] \\
&\le \sum_{j < i} \expect\left[ \frac{|\D_j|}{\nD/j\Cgap}\expect[g_j \mid \I_j ] \right] 
= \frac{\Cgap}{\nD} \expect\Big[ \sum_{j < i} |\D_j|j g_j \Big],
\end{align*}
where for each $j$, the outer expectation is taken over $\I_j$ and the inner expectation is over $v_j$.
The inequality follows from Lemma~\ref{lemma:imp}.

Let $\eta = \arg \max_j j g_j$ be the index yielding the highest summand. 
Since $g_j$ is non-increasing in $j$ by Lemma~\ref{lemma:nonincgain}, we have 
\[
\sum_{j < i} |\D_j|j g_j
\le \sum_{j < i} |\D_j| \eta g_\eta
\le \nD \eta g_\eta
\le \nD \sum_{j \le \eta} g_j
\le \nD \sum_{j < i} g_j.
\]
Therefore, we have
\begin{align*}
\frac{\Cgap}{\nD} \expect\Big[ \sum_{j < i} |\D_j|j g_j \Big]
\le \frac{\Cgap}{\nD} \expect\Big[ \nD \sum_{j < i} g_j \Big]
= \Cgap \expect[f(\I_i)].
\end{align*}
Combining the three inequalities proves that
$\expect[f(\I'_i)] \ge \expect[f(\I_i)] - \Cgap \expect[f(\I_i)]$,
completing the proof.
\end{proof}

The next lemma follows immediately.

\begin{lemma}
\label{lemma:robust-G}
Let $S$ be a set of items. Then for any $i$,
$$
    \expect[f(\I_i' \cup S)] \ge (1-\Cgap) \expect[f(\I_i \cup S)].
$$
\end{lemma}
\begin{proof}
\deferredproof
\end{proof}

Finally, we are ready to prove Theorem \ref{theorem:offline-matroid}.

\begin{proof}[Proof of Theorem~\ref{theorem:offline-matroid}]
Let $i$ be the largest index used by  Algorithm~\ref{alg-offline-coreset}, and write
let $\I = \I_{i+1}$ be the maximal partial solution in Algorithm~\ref{alg-offline-coreset}. Write also 
$\I' = \I \setminus \D$.
Similarly, \R is our coreset and $\R' = \R \setminus \D$.
To prove the claim, we compare \I with \OPT.
\begin{align*}
&f(\OPT)
\le f(\I \cup \OPT) 
\le f(\I) + f(\OPT \setminus \I \mid \I) \\
&\le f(\I) + f((\OPT \setminus \I) \cap \R' \mid \I) + f((\OPT \setminus \I) \setminus \R' \mid \I) \\
&\le f(\I) + (\nm+1) f(H) + f(\OPT \setminus \R' \mid \I).
\end{align*}

The last step is because any feasible solution in $\R'$, including $(\OPT \setminus \I) \cap \R'$, is within $\nm+1$ approximation of the greedy solution $H$ of Algorithm~\ref{alg-offline-selection} \citep{fisher1978analysis,calinescu2011maximizing}.
Now we deal with the last term.
Note that $\OPT \setminus \R' = \OPT \setminus \R$. Then
\begin{align*}
f(\OPT \setminus \R \mid \I) 
\le \sum_{u \in \OPT \setminus \R} f(u \mid \I) 
= \sum_{u \in O} f(u \mid \I),
\end{align*}
where $O = \{u \in \OPT \setminus \R : \I+u \not\in \I \}$.
Here the last step is due to the fact that the chosen \I is maximal, and an item will be discarded only when it is infeasible to \I.

Let $O = u_1, \ldots, u_{|O|}$ be the order in which Algorithm~\ref{alg-offline-coreset} discards the items in $O$.
Define a function $\pi$ with $\pi(u_\ell) = \lceil \ell / \nm \rceil$.
Let
\[
    O_j = \{ u \in O: \I_{j + 1} +u \notin \I\}.
\]
Note that $O_{j},\I_{j+1} \in \I$ and $\I_{j+1} + u \notin \I$ for every $u \in O_{j}$. Thus
$\I_{j+1}$ is a maximal independent set in $Y=\I_{j+1} \cup O_{j}$, and, by definition of \nm-system,
$|O_{j}| \le \nm |\I_{j+1}| = \nm j$ for all $j$.

Let $u = u_\ell \in O_{j} \setminus O_{j - 1}$.
Then, $\ell \leq \nm j$ and $\pi(u_\ell) \leq j$. Since $u$ is discarded after $v_j$ is added,
\[
	f(u \mid \I) \leq f(u \mid \I_j) \leq  f(v_j \mid \I_j) \leq  f(v_{\pi(u)} \mid \I_{\pi(u)}).
\]
Lastly, we have 
\begin{align*}
	\sum_{u \in O} f(u \mid \I) 
	 &\leq \sum_{u \in O} f(v_{\pi(u)} \mid \I_{\pi(u)}) 
	 \leq \nm \sum_{j} f(v_j \mid \I_j)
	 = \nm f(\I).
\end{align*}
Putting everything together, we have
\begin{align*}
f(\OPT) &\le f(\I) + (\nm+1) f(H) + \nm f(\I) \\
&\le (\nm+1+\frac{\nm+1}{1-\Cgap}) \expect[f(\ALG)],
\end{align*}
where the last step is due to Lemma \ref{lemma:robust}.

To bound the coreset size, note that $|\R|$ is bounded by
\[
	\nD + \sum_{j = 1}^\nI |C_j| \leq \nD + \sum_{j = 1}^\nI \max \left\{ 1, \frac{\nD}{\Cgap j}\right\} \leq \nD + k + \nD (\ln(\nI)+1) /\Cgap,
\]
completing the proof.
\end{proof}


\section{Experiments}
\label{section:experiment}

\begin{table}[t]
  \caption{Datasets statistics}
  \label{tbl:datasets}
  \centering
\begin{tabular}{lrrr}
\toprule
Dataset & $\nV = |\V|$	&$\nI$	&\M  \\
\midrule
Movielens \cite{harper2015movielens}	&22 046	&20	&2-matroid	\\
Facial images \cite{zhang2017age}	&23 705	&25	&1-matroid	\\
Github social network \cite{rozemberczki2019multiscale}	&37 700	&20	&cardinality	\\
Uber pickups \cite{uber2020}	&50 000	&25	&1-matroid	\\
Songs \cite{Bertin-Mahieux2011}	&137 543	&20	&cardinality	\\
\bottomrule
\end{tabular}
\end{table}

\begin{figure}[t]
    \centering
    \subcaptionbox{Movie recommendation}{
    \includegraphics[width=.49\textwidth, trim = 2.5cm 2cm 0.8cm 0.7cm, clip]{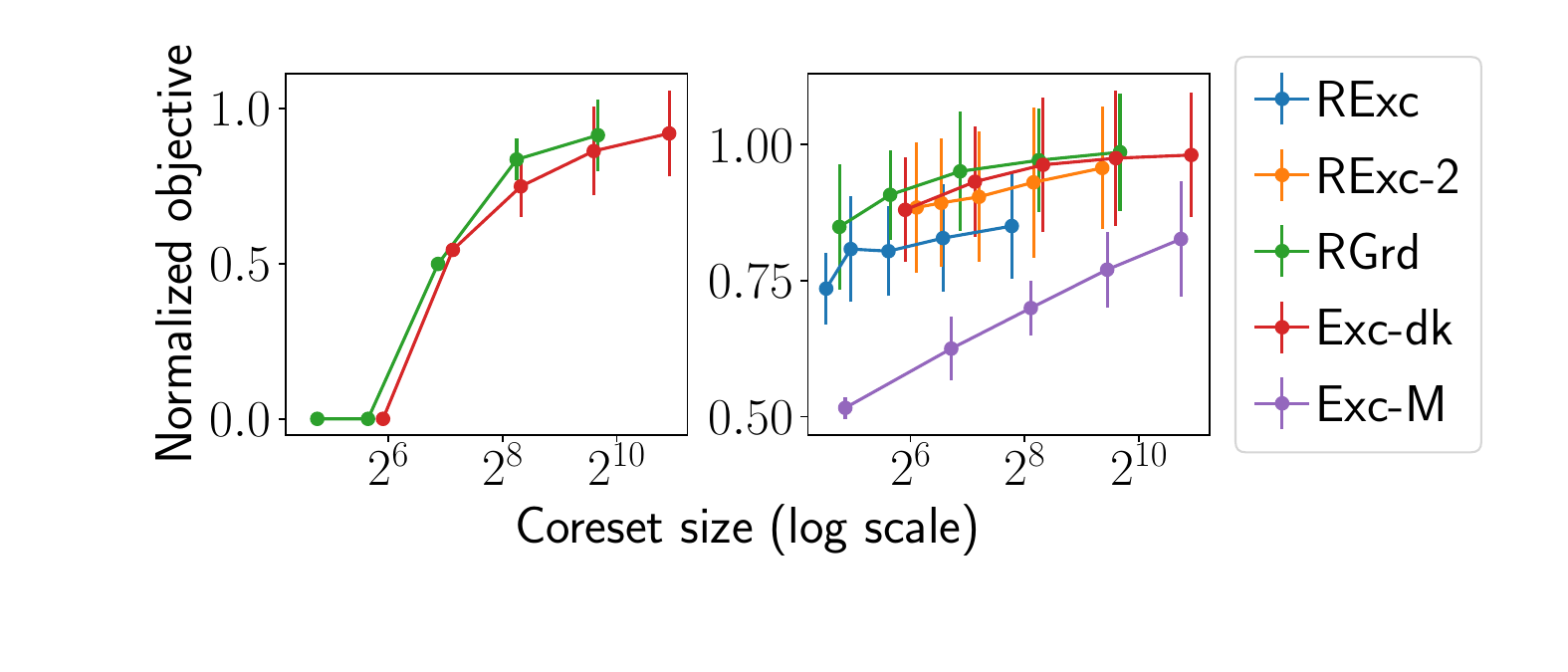}
    }
    \subcaptionbox{Facial image selection (x-axis: coreset size)}{
    \includegraphics[width=.49\textwidth, trim = 2.5cm 3cm 0.8cm 0.7cm, clip]{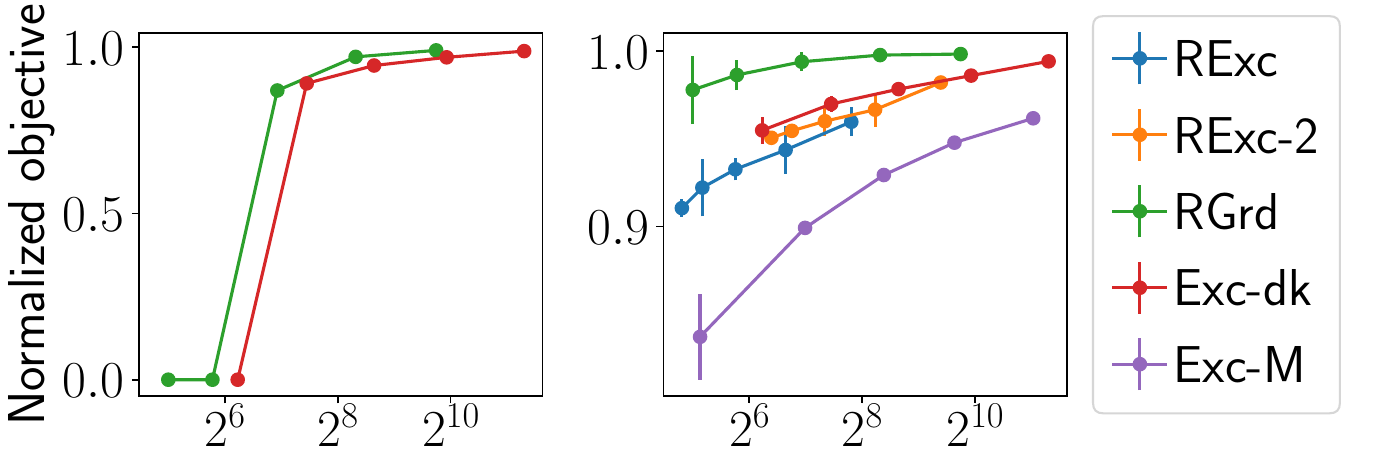}
    }
    \subcaptionbox{Uber pickups summarization (x-axis: coreset size)}{
    \includegraphics[width=.49\textwidth, trim = 2.5cm 3cm 0.8cm 0.7cm, clip]{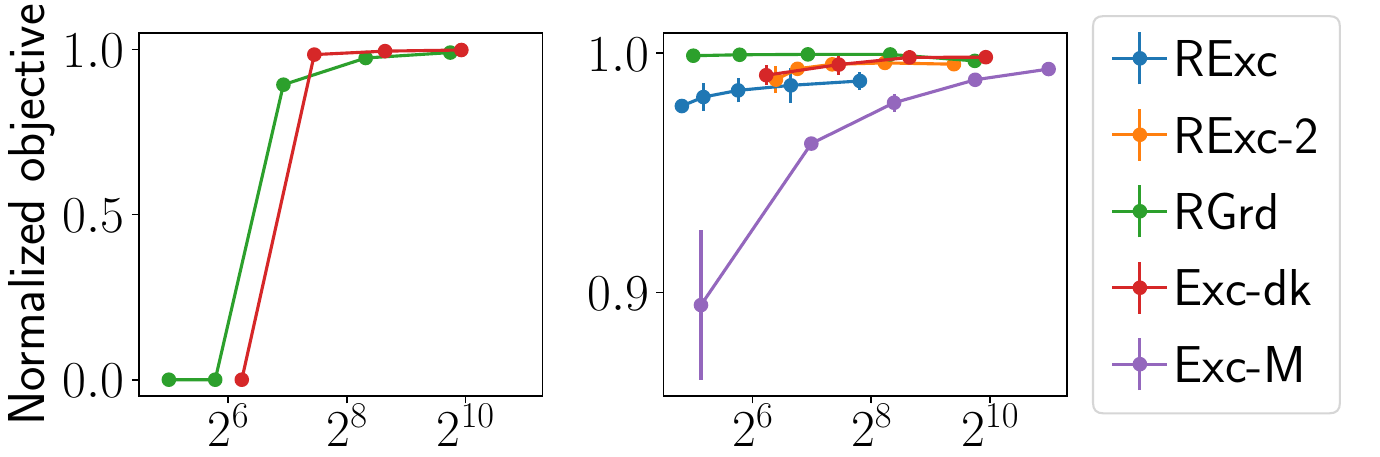}
    }
	\subcaptionbox{Network influence maximization in Github (x-axis: coreset size)}{
    \includegraphics[width=.49\textwidth, trim = 2.5cm 3cm 0.8cm 0.7cm, clip]{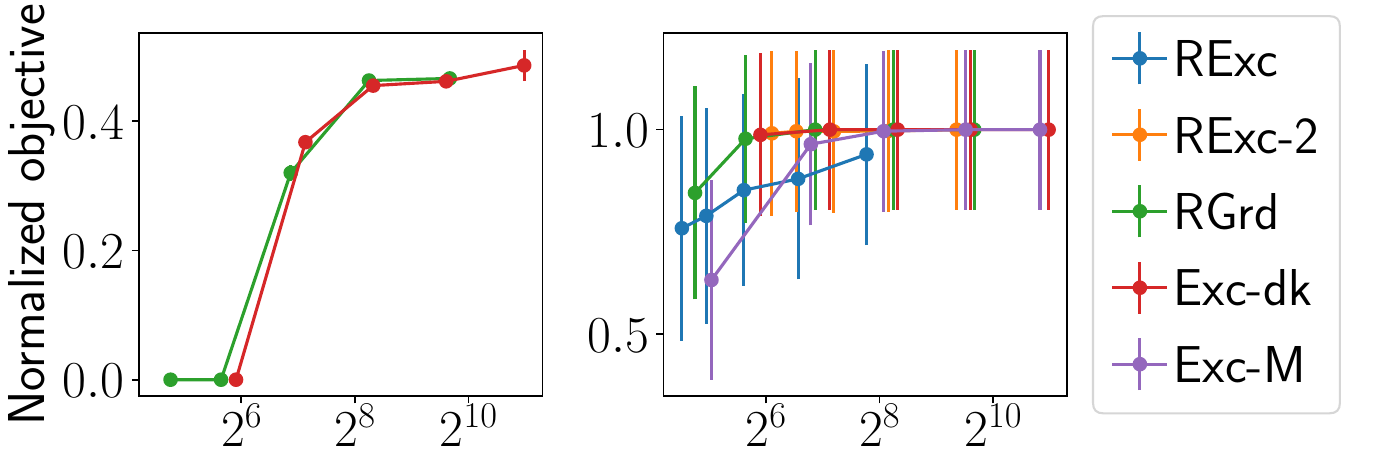}
    }
    \subcaptionbox{Popular song selection}{
    \includegraphics[width=.49\textwidth, trim = 2.5cm 2cm 0.8cm 0.7cm, clip]{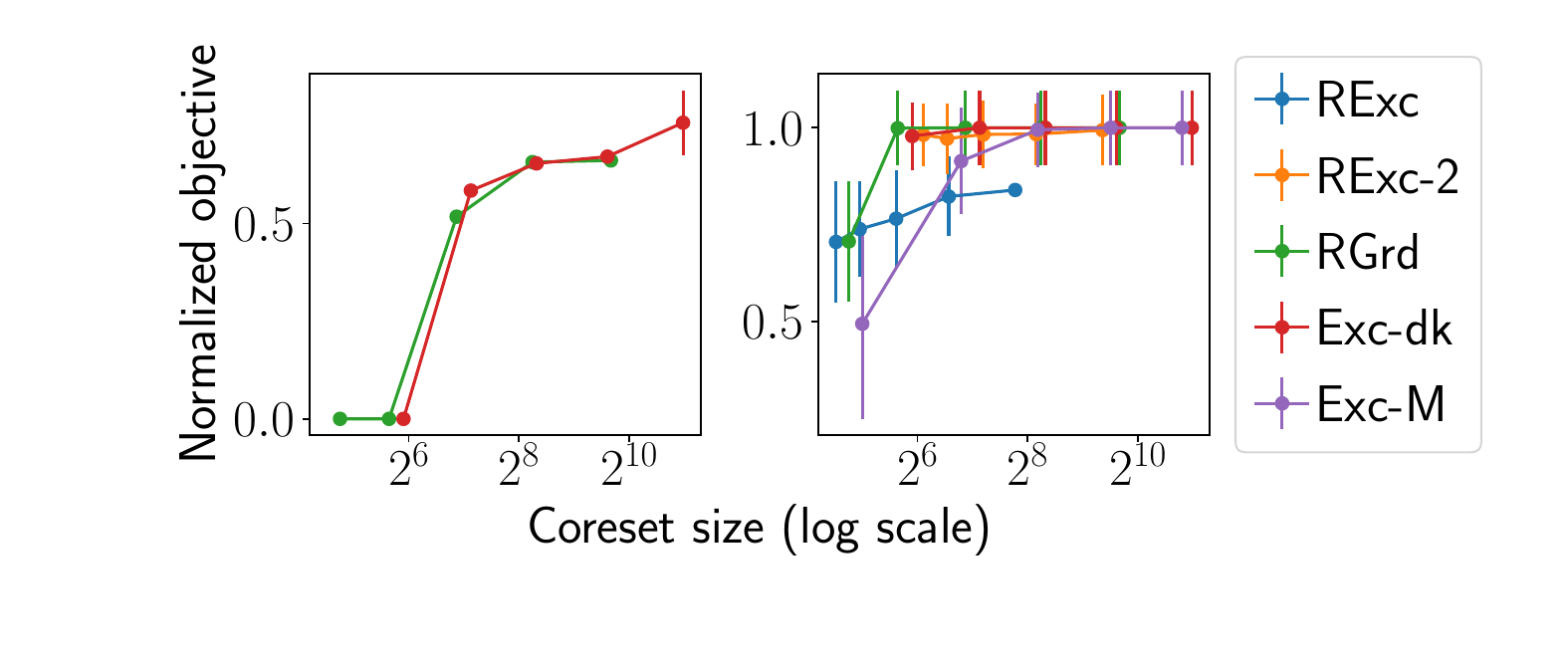}
    }
    \caption{\label{fig:results}
    Experiment results.
    The adversary (left: adaptive, right: static) deletes items of a fixed size 100.
    Parameter \nD (and coreset size) in each algorithm is gradually increased to 100.
    }
\end{figure}

In this section, we evaluate the proposed algorithms against state-of-the-art baselines.
All methods are tasked with various subset-selection applications over real-life data.
Statistics of the datasets used are summarized in Table~\ref{tbl:datasets}.
The applications are described below 
(Sections~\ref{section:experiment-movie}--\ref{section:experiment-song})
followed by a discussion of experimental results (Section~\ref{section:result}) and 
an evaluation of running time (Section~\ref{section:runtime}).
Further details of the experiments are deferred to Section~\ref{section:experiment-details}~\cite{appendix}.
We introduce the competing algorithms and adversaries below.

\paragraph{Algorithms}
Competing algorithms include:
\begin{itemize}
	\item \algrexc, the robust \exc algorithm presented in Algorithm~\ref{alg-exc}, 
    Section~\ref{section:streaming};
    \item \algcrexc, two cascading instances of the \algrexc algorithm;
	\item \alggreedy, the offline robust greedy algorithm presented in Algorithm~\ref{alg-offline-coreset}, Section~\ref{section:offline};
	\item \algdk, a flexible reduction proposed by \citet{mirzasoleiman2017deletion} that constructs $\nD+1$ cascading \exc instances;
	\item \algdutting, the previous state-of-the-art robust \exc algorithm by \citet{dutting2022deletion} that performs uniform sampling on top of multiple candidate sets, each associated with an increasing threshold on marginal gain.
\end{itemize}
Their objective values of all methods 
are normalized by that of an \emph{omniscient} greedy algorithm,
which is aware of deleted items in advance.
Algorithms \alggreedy and \algdk are also challenged to an adaptive adversary.
A fixed parameter $\Cgap=0.5$ is used to avoid large coresets.

\paragraph{Adversary}
We consider two types of adversaries, static and adaptive, 
which make deletions over the whole universe of items \V or only over the coreset, respectively.
To introduce randomness in a principled way,
given an integer \nD, we simulate an adversary by running 
the \emph{Stochastic Greedy} algorithm~\cite{mirzasoleiman2015lazier} 
and obtain a deletion set~\D of size \nD.
Concretely, in each iteration, we add into \D the greedy item among a multiple of $z/\nD$ random items,
where $z=\nV$ for a static adversary and $z$ is the coreset size for an adaptive one.

As a general strategy, we let the adversary delete $100$ items, and
we gradually increase the parameter \nD in each algorithm until it reaches $100$.

\subsection{Personalized movie recommendation}
\label{section:experiment-movie}
Robust recommendation is favorable in practice due to uncertain deletions caused by user preference.
A popular approach to personalized recommendation \cite{mitrovic2017streaming} is to optimize the following submodular function,
\[
	f_u(S) = (1-\lambda) \sum_{v \in S} \text{sim}(u,v) + \frac{\lambda \nI}{|\V|} \, \sum_{w \in \V} \max_{v \in S} \text{sim}(w,v),
\]
such that $|S| \le \nI$.
Here $\text{sim}(u,v)$ measures the relevance of an item $v$ to the target user $u$, and
$\text{sim}(w,v)$ the similarity between two items $w,v$.
The second term represents a notion of representative\-ness,
i.e., for every non-selected item $w \in \V$, 
there exists some item $v \in S$ that is similar enough to $w$.

We choose the Movielens dataset \cite{harper2015movielens},
which consists of 9\,724 movies and hundreds of users.
We obtain feature vectors for users and movies by applying SVD on the user-movie rating matrix, and
let $\text{sim}(\cdot,\cdot)$ be the natural dot product.
A random user is chosen for the recommendation task.
A movie may belong to more than one of 20 genres, and
we further impose a 2-matroid on a feasible solution $S$, i.e.,
every movie can be selected at most once and at most one movie can be selected for each genre.
Tradeoff parameter $\lambda$ is fixed to 0.5.
The results are reported in Figure~\ref{fig:results}(a).

\subsection{Facial image selection}
\label{section:face}

Exemplar-based applications, such as nearest-neighbor models and recommender systems, 
are ubiquitous in data science.
However, the ``right to be forgotten'' can lead to the case where some items must be be deleted \citep{voigt2017eu}.
In such cases, a robust coreset is desirable, 
so as to maintain a representative summary for applications after data-item deletions.

A dataset \V can be summarized by a representative subset of data $S$ 
via minimizing the classic \nI-medoid function,
\[
	g(S) = \sum_{v \in \V} \min_{u \in S} d(u,v), 
\]
where $d(u,v)$ measures the distance between $u$ and $v$.
Intuitively, 
for each item $v$ in the data, there should exist some item $u$ in the summary $S$ that is close to $v$.
The above function can be turned into a submodular maximization problem 
by measuring the total reduction of distance with respect to some item $w \in V$ instead, i.e.,
$f(S) = g(\{w\}) - g(S+\{w\})$ \cite{mirzasoleiman2013distributed}.
We let $w$ be the an arbitrary random item.

We experiment with a dataset of
facial images~\cite{zhang2017age},
under a partition matroid according to races (5 images per race and $\nI=25$),
with the distance function being the $\ell_1$ metric.
The results are reported in Figure~\ref{fig:results}(b).

\subsection{Geolocation data summarization}
\label{section:experiment-geolocation}

We experiment with a similar task as in Section~\ref{section:face}, 
except for a different dataset,
Uber pickups \cite{uber2020}.
Every data point indicates a location of Uber pickups in New York City in April, 2014.
We measure the distance by a natural $\ell_1$ metric.
A partition matroid is imposed according to the base companies 
(at most 5 pickups for each company and $\nI=25$).
The results are reported in Figure~\ref{fig:results}(c).

\subsection{Network influence maximization}
\label{section:experiment-influence}

For viral-marketing applications in social networks,
the goal is to identify a small set of seed nodes who can influence many other users.
For popular diffusion models, the number of influenced nodes is a submodular function
of the seed set \cite{kempe2015maximizing}.
Here, we consider deletion-robust viral marketing for a simple diffusion model, 
where a seed node always influences all its neighbors, i.e.,
$f(S) = |\cup_{v \in S} N(v)|$ returns a dominating set, where $N(v)$ represents neighbors of~$v$.
We choose the dataset of Github social network \cite{rozemberczki2019multiscale}, and
specify a cardinality limit of $\nI=20$.
The results are reported in Figure~\ref{fig:results}(d).

\subsection{Popular song selection}
\label{section:experiment-song}

Given song-by-song listening history of users, 
one wishes to select a set of popular songs $S$ that can ``cover'' the most users.
A user is covered if she likes at least one song in $S$.
That is, 
$f(S) = |\cup_{v \in S} L(v)|$,
where $L(v)$ represents the set of users who like song~$v$.
We aim for a deletion-robust coreset for such popular songs.
Concretely, we use the million song dataset \cite{Bertin-Mahieux2011}, 
consisting of triples representing a user, song, and play count.
We assume that a user likes a song if the song is played more than once.
We impose a cardinality limit of $\nI=20$.
The results are reported in Figure~\ref{fig:results}(e).

\subsection{Discussion of results}
\label{section:result}

Overall, 
against a static adversary,
the proposed \alggreedy algorithm performs the best and converges 
with the smallest coreset,
while the proposed \algrexc algorithm achieves relatively good performance while requiring the most parsimonious coreset.
Note that the size difference in the coreset will become more extreme as $\nI$ increases.
For an adaptive adversary,
the \alggreedy algorithm remains the most robust.

The \algdutting algorithm has the worst performance most of the time, except for tasks with a simple cardinality constraint and a relatively large coreset size.
This behavior illustrates the insufficient efficacy of uniform sampling.
Uniform sampling on top of thresholded candidate sets,
which is adopted by many previous robust algorithms \cite{kazemi2018scalable,dutting2022deletion},
faces a dilemma between a large coreset or a crude distinction of item importance 
(i.e., few crude thresholds due to large $\Cgap$).
This issue is properly addressed by the non-uniform sampling technique in this paper.

On the other hand, cascading instances in the \algdk algorithm appears to be another 
promising way for preserving valuable items in stream computation.
However, this approach comes with a cost of expensive computation 
(see Section~\ref{section:runtime}).
Besides, its coreset size explodes even with a moderate value of~\nD, 
Staying with a small \nD parameter, however, fails to secure a theoretical guarantee when more items are deleted.

Algorithms \algrexc an \algdk preserve valuable and compatible items in two different ways.
This naturally suggests that one can combine the best of both worlds
by constructing a small number of cascading \algrexc instances.
Then one is expected to further enhance the performance while maintaining a parsimonious coreset and a strong guarantee.
This is indeed the case as reflected by the remarkable performance of the \algcrexc algorithm, which uses merely two instances of \algrexc.

In summary, 
we conclude that the \alggreedy algorithm is a reliable choice if an offline algorithm is allowed.
In a streaming setting,
a small number of cascading \algrexc instances is recommended.

\subsection{Running time analysis}
The running time of all algorithms over the song dataset is shown in Figure~\ref{fig:runtime}.
The most significant message of Figure~\ref{fig:runtime} is that 
the \algdk algorithm is computationally costly when the value of parameter $\nD$ grows.

\label{section:runtime}
\begin{figure}[t]
    \centering
    \includegraphics[width=.35\textwidth, trim = 0cm 0.8cm 0cm 0cm, clip]{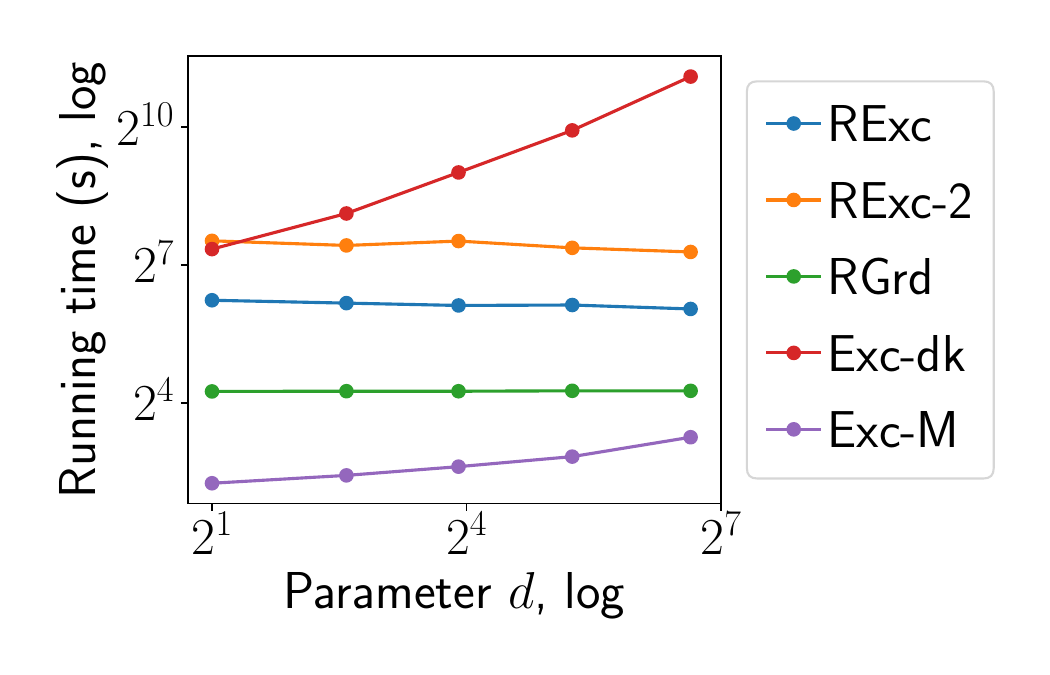}
    \caption{Running time on the song dataset
    }
    \label{fig:runtime}
\end{figure}

\section{Conclusion}
\label{section:conclusion}

In the presence of adversarial deletions up to \nD items,
we propose a single-pass streaming algorithm that 
yields $(1-2\Cgap)/(4\nm)$-approximation for maximizing a non-decreasing submodular function 
under a general \nm-matroid constraint and 
requires an (asymptotically) optimal coreset size $\nI + \nD/\Cgap$, where \nI is the maximum size of a feasible solution.
Besides, we develop an offline greedy algorithm 
that guarantees stronger approximation ratios, and
performs effectively even against an adaptive adversary.

One vital tool for robustness in the proposed algorithms is ``uselessness'' sampling
that preserves valuable items within the candidate set and avoids great loss caused by adversarial deletions in expectation.
Another insight is a close connection between robustness and streaming algorithms.
The latter ensures a quality guarantee given an arbitrary arrival order of items, 
including the specific random order introduced by the sampling.

Potential directions for future work include 
a potentially stronger approximation ratio in the offline setting,
extensions to non-monotone submodular maximization, and 
a stronger adaptive adversary.

\todo{Extend to non-monotone $f$: seems difficult, as a dynamic (instead of fixed) \val is needed \citep{chekuri2015streaming}.}

\section*{Acknowledgment}
%
This research is supported by the Academy of Finland projects MALSOME (343045), AIDA (317085) and MLDB (325117),
the ERC Advanced Grant REBOUND (834862), 
the EC H2020 RIA project SoBigData++ (871042), 
and the Wallenberg AI, Autonomous Systems and Software Program (WASP) 
funded by the Knut and Alice Wallenberg Foundation.

\bibliographystyle{IEEEtranN} 
\bibliography{IEEEabrv,references}		
\fi 

\ifsupp 
\newpage
\appendix

\subsection{Omitted proofs}
\label{section:omitted-proofs}

\begin{proof}[Proof of Lemma~\ref{lemma:nonincgain}]
Note that $v_j$ is sampled from $\C_j$ and $v_{i}$ is sampled from~$\C_{i}$, and $\C_{j} \cap \C_{i} = \emptyset$.
Since $\C_j$ includes top items sorted by $f(\cdot \mid \I_j)$ and $v_{i} \not\in \C_j$, we have
$f(v_j \mid \I_j) \ge f(v_{i} \mid \I_j) \ge f(v_{i} \mid \I_{i})$, 
where the second inequality is due to submodularity.
\end{proof}

\begin{proof}[Proof of Lemma~\ref{lemma:robust-G}]
Since $\I_i' \subseteq \I_i$, submodularity and Lemma \ref{lemma:robust} imply that
\begin{align*}
    \expect[f(\I_i \cup S)] - \expect[f(\I_i' \cup S)] & \leq \expect[f(\I_i)] - \expect[f(\I_i')] \\
    & \leq \Cgap \expect[f(\I_i)]  \leq \Cgap \expect[f(\I_i \cup S)],
\end{align*}
proving the claim.
\end{proof}

\subsection{Streaming algorithm for the \rcc problem}
\label{section:streaming-cardinality}

\begin{algorithm}[t]
\DontPrintSemicolon
\KwIn{parameter \Cgap}
$\topd \gets$ the set of top-\nD singletons seen so far\;
$\topv_\nD \gets$ the value of the top $(\nD+1)$-th singleton \;
\For{each new item $v$ in the stream}{
	Update $\topv_\nD$ and $\topd$\;
	\If{\topd is changed}{
		$v \gets$ the swapped-out item in $\topd$\;
	}
	$\thrs \gets \{ (1+\epsilon)^i : \frac{\topv_\nD}{2\nI(1+\Cgap)} \le (1+\epsilon)^i \le \topv_\nD, i \in \naturals \}$\;
	\For{each $\thr \in \thrs$ in parallel}{\label{step:innerloop}
		\If{$f(v \mid \I_\thr) \ge \thr$ and $|\I_\thr| < \nI$}{
			$\C_\thr \gets \C_\thr + v$\;
		}
		\If{$|\C_\thr| \ge \nD/\Cgap$}{
			Sample an item $v$ from $\C_\thr$ with a probability proportional to $1 / f(v \mid \I_\thr)$\;
			$\I_\thr \gets \I_\thr + v$\;
			$\C_\thr \gets \{ v \in \C_\thr: f(v \mid \I_\thr) \ge \thr \}$\;
		}
	}
}
$\R \gets \topd \cup \left(\bigcup_{\thr \in \thrs} \C_\thr \right) \cup \left(\bigcup_{\thr \in \thrs} \I_\thr \right)$\;
\Return{$\R, \left\{\I_\thr \right\}_{\thr \in \thrs}$}
\caption{Streaming robust coreset for \rcc}\label{alg-cardinality-streaming-coreset}
\end{algorithm}

A well-known technique developed by \citet{badanidiyuru2014streaming} enables a one-pass streaming algorithm for non-robust submodular maximization.
That is, the algorithm is restricted to read items in \V only once with very limited memory.
The key is to make multiple guesses at the threshold $\thr^*$ such that
$\thr^* \le \frac{f(\OPT)}{2\nI} \le (1+\Cgap) \thr^*$, 
and then build a candidate solution for each guessed threshold in parallel.
The guesses depend on the top singleton encountered so far, and are dynamically updated along the process.

As pointed out in \citet{kazemi2018scalable}, a natural extension for the robust setting is to make guesses according to the set of top $\nD+1$ singletons we have seen so far.
To be more specific, we dynamically maintain a set of geometrically-increasing thresholds within range $[\topv_\nD/2\nI, \topv_\nD]$, 
where $\topv_\nD$ is the value of the current $(\nD+1)$-th top singleton.
Besides, to cope with a static adversary, every item in a candidate solution is randomly sampled from a candidate set of large size.

We adopt the same approach as \citet{kazemi2018scalable}, but significantly simplify their algorithm and improve the coreset size, thanks to the importance sampling technique in Lemma~\ref{lemma:imp}.

We prove Theorem \ref{theorem:streaming-cardinality} in the rest of this section, that is,
Algorithm \ref{alg-cardinality-streaming-coreset} and \ref{alg-offline-selection-cardinality}
yield $(1-2\Cgap)/2$ approximation guarantee for the \rcc problem, with a coreset size $\bigO((\nD/\Cgap + \nI) \log(\nI)/\Cgap)$.
Algorithm~\ref{alg-cardinality-streaming-coreset} constructs the coreset in one-pass using
$\bigO\left(\frac{\log(\nI)}{\Cgap} \big(\nV + \nI \frac{\nD}{\Cgap}\big)\right)$ queries.
Algorithm~\ref{alg-offline-selection-cardinality} needs to be slightly modified when receiving the coreset from Algorithm~\ref{alg-cardinality-streaming-coreset}.
Specifically, at Step~\ref{step:I}, $\I_\thr$ is given directly.

The proof is very similar to that of Theorem \ref{theorem:offline-cardinality},
except that we replace Lemma \ref{lemma:gooditems} and \ref{lemma:robust} with the following two new lemmas.


\begin{lemma}
\label{lemma:gooditems-streaming}
Let $\thr \in \thrs$ be a threshold, $\I_\thr$ its associated partial solution, and $\C_\thr$ candidate set.
If $|\I_\thr| < \nI$, 
then for every item $v \in \V \setminus (\C_\thr \cup \topd)$, 
we have
$f(v \mid \I_\thr) < \thr$.
\end{lemma}
\begin{proof}
Select $\thr \in \thrs$ and $v \in \V \setminus (\C_\thr \cup \topd)$.
Since $v \notin \topd$, consider the iteration when $v$ is properly processed either due
to being a new item or being an item leaving $\topd$. 
Let $\topv_\nD'$, $\C_\thr'$, $\I_\thr'$, and $\thrs'$ be the variables of Algorithm~\ref{alg-cardinality-streaming-coreset} right before the inner for-loop (Step~\ref{step:innerloop}).

Assume that $\tau \in \thrs'$.
If $f(v \mid \I_\thr) \geq \thr$, then $f(v \mid \I_\thr') \geq \thr$ and
$v$ is added to $\C'_\tau$ and never filtered out, that is, $v \in \C_\tau$
which is a contradiction. Thus, $f(v \mid \I_\thr) < \thr$.

If $\tau \notin \thrs'$, then $\thr > \topv_\nD'$ since $\topv_\nD$ can only increase.
Consequently, $f(v \mid \I_\thr) \leq f(v) \leq \topv_\nD' < \thr$.
\end{proof}

\begin{lemma}\label{lemma:robust-streaming}
For each threshold \thr and its associated partial solution $\I_\thr$ kept by Algorithm \ref{alg-cardinality-streaming-coreset}, 
we have
$\expect[f(\I_\thr')] \ge (1-\Cgap) \expect[f(\I_\thr)]$,
where $\I_\thr' = \I_\thr \setminus \D$.
\end{lemma}
\begin{proof}
We fix an arbitrary threshold \thr, and write $\I = \I_\thr = \{v_1, \ldots, v_i\}$.
Let us write $\I_j =  \{v_1, \ldots, v_{j-1}\}$,
and define $\C_j$ to the candidate set $\C_\thr$ from which we sample $v_j$. Then
\begin{align*}
\expect[f(\I')] 
& = \expect\Big[\sum_{j \leq i} f(v_j \mid \I_j \setminus \D) \indicator[v_j \notin \D]\Big] \\
&\ge \expect\Big[\sum_{j \leq i} f(v_j \mid \I_j) \indicator[v_j \notin \D]\Big] \\
&= \expect[f(\I)] - \sum_{j \leq i} \expect\Big[f(v_j \mid \I_j) \indicator[v_j \in \D]\Big] \\
&\ge \expect[f(\I)] - \sum_{j \leq i} \frac{|\D|}{|\C_j|} \expect[f(v_j \mid \I_j)]
&& \triangleright\text{Lemma \ref{lemma:imp}} \\
&\ge \expect[f(\I)] - \Cgap \sum_{j \leq i} \expect[f(v_j \mid \I_j)] 
&& \triangleright |\C_j| \geq \nD/\Cgap \\
&= (1-\Cgap) \expect[f(\I)],
\end{align*}
completing the proof.
\end{proof}

\begin{proof}[Proof of Theorem~\ref{theorem:streaming-cardinality}]
The proof for approximation is essentially the same as in the proof of Theorem~\ref{theorem:offline-cardinality},
except that Lemma~\ref{lemma:gooditems} is replaced with Lemma~\ref{lemma:gooditems-streaming}
and Lemma~\ref{lemma:robust} is replaced with Lemma~\ref{lemma:robust-streaming}.
We omit the details to avoid repetition.

We complete the proof by calculating the coreset size.
Note that during each iteration $|C_\tau|$ can increase only by 1.
If after addition $|C_\thr| > \nD/\Cgap$, then $v$ is sampled from $C_\tau$ and will be filtered out since $f(v \mid \I_\thr) = 0$.
Thus, in the end $|C_\thr| \leq \nD/\Cgap$.
There are $\bigO(\log(\nI)/\Cgap)$ different thresholds, and for each threshold \thr we keep at most $\nD/\Cgap$ items in $\C_\thr$ and a partial solution $\I_\thr$ of at most size \nI.
Hence, the total coreset size is $\bigO((\nD/\Cgap + \nI) \log(\nI)/\Cgap)$.
\end{proof}

\subsection{Offline algorithm for the \rcc problem}
\label{section:offline-cardinality}

\begin{algorithm}[t]
\DontPrintSemicolon
\KwIn{Coreset and auxiliary information $(\R, \{\I_j\}_j)$ returned by Algorithm~\ref{alg-offline-coreset}, set of deleted items $\D$, parameter \Cgap}
$\R' \gets \R \setminus \D$\; 
$\topv \gets $ value of the top singleton in $\R'$ according to $f(\{v\})$\;
$\thrs \gets \left\{ (1+\epsilon)^i : \frac{\topv}{2\nI(1+\Cgap)} \le (1+\epsilon)^i \le \topv, i \in \naturals \right\}$ \label{step:T}\;
\For{$\thr \in \thrs$}{
	$\I_\thr \gets \I_{j + 1}$, where $j \gets \max \left\{j : f(\I_{j+1} \setminus \I_j \mid \I_j) \ge \thr\right\}$ \label{step:I}\;
	$\I_\thr' \gets \I_\thr \setminus \D$\;
	\For{$v \in \R'$}{
		\If{$f(v \mid \I_\thr') \ge \thr$ and $|\I_\thr'| < \nI$}{
			$\I_\thr' \gets \I_\thr' + v$\;
		}
	}
} 
\Return{the best solution among $\{\I_{\thr}'\}_{\thr \in \thrs}$}\;
\caption{Construction of \rcc solution after deletion}
\label{alg-offline-selection-cardinality}
\end{algorithm}

In the case of cardinality constraint, we can obtain a stronger bound than in the general case,  
by executing a different Algorithm~\ref{alg-offline-selection-cardinality} after receiving a coreset from Algorithms~\ref{alg-offline-coreset}.
Algorithm~\ref{alg-offline-selection-cardinality} is inspired by the \sieve algorithm in \citet{badanidiyuru2014streaming}, 
which makes multiple guesses at the threshold $\thr^*$ such that
$\thr^* \le \frac{f(\OPT)}{2\nI} \le (1+\Cgap) \thr^*$, 
and then builds a candidate solution for each guessed threshold in parallel.

Let us write $\I_\thr = \I_{i + 1}$ and $\I_\thr' = \I_{i + 1}'$, 
where $i = \max \{ j : f(v_j \mid \I_j) \ge \thr\}$ 
is the subset of the tentative solutions built in Algorithms~\ref{alg-offline-coreset} and~\ref{alg-offline-selection-cardinality} filtered by a threshold $\thr$.
The next step is to show that we do not miss in our coreset any feasible item~$v$ with marginal gain $f(v \mid \I_{\thr}) \ge \thr$ among items in \V.

\begin{lemma}
\label{lemma:gooditems}
Assume $\thr > 0$ and let $i = \max \{ j : f(v_j \mid \I_j) \ge \thr \}$.
Shorten $\I = \I_{i + 1}$.
Let $v$ be an item that can be added to $\I$ with a gain of at least $\thr$,
that is, $f(v \mid \I) \ge \thr$ and $\I + v \in \M$.
Let \R be the coreset returned by Algorithm~\ref{alg-offline-coreset}.
Then $v$ is a member of \R.  
\end{lemma}
\begin{proof}
If $\I$ is the last set in the loop of Algorithm~\ref{alg-offline-coreset}, then 
there is nothing to prove, as by definition
there are no items that can be added to $\I$ without violating $\M$.
Thus, we assume $\I$ is not the last set and $v_{i + 1}$ and $C_{i + 1}$ exist.

By definition of $i$, we have $f(v_{i + 1} \mid \I) < \thr$.
Since $f(v \mid \I) \ge \thr$ either $v \in C_{i + 1}$ or
$v$ has been removed earlier, that is $v \in C_j$ for $j \leq i$.
In either case, $v$ is in \R.  
\end{proof}

The solutions $\I'_\thr$ constructed by Algorithm~\ref{alg-offline-selection-cardinality} are of form $\I_i' \cup S$.
By Lemma~\ref{lemma:robust-G}, we already know that $f(\I_i' \cup S)$ is close to $f(\I_i \cup S)$ in expectation, so we are free to study
$\I_i \cup S$ instead of $\I_i' \cup S$.
Next we will bound the former using a thresholding argument. Similar arguments have been used by 
\citet{badanidiyuru2014streaming} and \citet{kazemi2018scalable}.
\begin{lemma}
\label{lemma:good-G}
For any $\thr$, write $S_\thr$ the additional items added to $\I'_\thr$ by Algorithm~\ref{alg-offline-selection-cardinality}.
Define $G_\thr = \I_\thr \cup S_\thr$.
Then
\[
f(G_\thr) \ge \min\left\{ f(\OPT) - \nI \thr, \nI \thr \right\}.
\]
\end{lemma}
\begin{proof}
For simplicity, let us shorten $G_\thr$, $G'_\thr$, $\I_\thr$, $\I_\thr'$, and $S_\thr$ with
$G$, $G'$, $\I$, $\I'$, and $S$, respectively.
We will prove the lemma by considering three cases.

Case 1: Assume $|\I| = \nI$. Then $f(G) \geq f(\I) \geq \nI \thr$.

Case 2: Assume $|G'| = \nI$. Then $f(G) \geq f(G') \geq \nI \thr$.

Case 3: Assume that $|\I| < \nI$ and $|G'| < \nI$.
The main idea is to show that every item in \OPT has a marginal gain less than \thr with respect to $G$.
Write
\begin{align*}
&f(\OPT) 
\le f(\OPT \cup G) 
\le f(G) + \sum_{v \in \OPT \setminus G} f(v \mid G) \\
&= f(G) + \sum_{v \in (\OPT \setminus G) \cap \R} f(v \mid G) + \sum_{v \in (\OPT \setminus G) \setminus \R} f(v \mid G) 
.
\end{align*}
We discuss the last two terms separately.

To bound the first term,
let $v \in (\OPT \setminus G) \cap \R$.
Since $\OPT \subseteq \V \setminus \D$, we have $v \in \R \setminus \D$.
If $f(v \mid G) \geq \tau$, then $v \in \I'$ or, since $|G'| < \nI$, the item $v$ is added to $S$ 
by Algorithm~\ref{alg-offline-selection-cardinality}. Thus, $f(v \mid G) < \tau$.

To bound the second term,
let $v \in (\OPT \setminus G) \setminus \R$.
Since $|\I| < \nI$, Lemma~\ref{lemma:gooditems} implies that $\thr > f(v \mid \I) \ge f(v \mid G)$.

Since $|\OPT| \leq \nI$, we have $f(\OPT) \le f(G) + \nI \thr$, proving the lemma.
\end{proof}

The next step is to show that there exists $\thr^* \in \thrs$ in thresholds \thrs enumerated by Algorithm~\ref{alg-offline-selection-cardinality} such that 
$\thr^* \le \frac{f(\OPT)}{2 \nI} \le (1+\Cgap) \thr^*$.
\begin{lemma}
\label{lemma:thr}
Let $\thrs$ be the set of thresholds enumerated by Algorithm~\ref{alg-offline-selection-cardinality}.
There exists $\thr^* \in \thrs$ such that $\thr^* \le \frac{f(\OPT)}{2 |\OPT|} \le (1+\Cgap) \thr^*$.
\end{lemma}
\begin{proof}
Let $\topv = \max_{v \in \R'} f(v)$ be the value of the top singleton in $\R'$.
Note that $R$ contains the top $\nD+1$ singletons as it contains
$\nD$ top singletons and $C_1$. Thus $\topv = \max_{v \in \V \setminus \D} f(v)$. 
Consequently, $\topv \le f(\OPT) \le |\OPT| \topv$.

Since $|\OPT| \le \nI$, 
$\frac{\OPT}{2|\OPT|}$ lies within the range $[\frac{\topv}{2\nI}, \topv]$.
An approximately close threshold $\thr^*$ can be found by enumeration in an exponential scale with a base $1+\Cgap$.
\end{proof}

Finally, we are ready to prove Theorem \ref{theorem:offline-cardinality}.
\begin{proof}[Proof of Theorem~\ref{theorem:offline-cardinality}]
Let $\tau^*$ be as given by Lemma~\ref{lemma:thr}.
Write~$G$ and $G'$ to be $G_{\tau^*}$ and $G'_{\tau^*}$
as given in Lemma~\ref{lemma:good-G}.

Lemmas~\ref{lemma:thr}~and~\ref{lemma:good-G} state that
\begin{align*}
	\expect[f(G)] & \geq \min \left\{ f(\OPT) - \nI \tau^*, \nI \tau^* \right\}  \\
	& \geq \min\left\{ f(\OPT) - f(\OPT) / 2, (1 - \Cgap) f(\OPT) / 2 \right\}   \\
	& \geq (1 - \Cgap) f(\OPT) / 2.
\end{align*}

Lemma~\ref{lemma:robust-G} states that
\begin{align*}
\expect[f(\ALG)]
 & \ge \expect[f(G')] \ge (1-\Cgap) \expect[f(G)] \\
 & \ge \frac{(1-\Cgap)^2}{2} f(\OPT) \ge \frac{1-2\Cgap}{2}f(\OPT),
\end{align*}
proving the approximation.

The bound on the coreset size is proved in the same way as in Theorem~\ref{theorem:offline-matroid}.
\end{proof}

\subsection{Further details of experiments}
\label{section:experiment-details}

Every algorithm returns the best between its solution and an additional greedy selection over the coreset after deletions.
The greedy selection possesses better quality most of the time.

In Figure~\ref{fig:results}, the error bar at each point is by three random runs.

In Section~\ref{section:experiment-movie}, 
A universe set of 9\,724 movies is expanded into a larger set of size 22\,046, 
by considering all movie-genre tuples.

\fi 

\end{document}